\documentclass[a4paper,twocolumn,11pt]{quantumarticle}
\pdfoutput=1
\usepackage[utf8]{inputenc}
\usepackage[english]{babel}
\usepackage[T1]{fontenc}
\usepackage{amsmath}
\usepackage{amssymb}
\usepackage{amsthm}
\usepackage{amsfonts}
\usepackage{hyperref}
\usepackage{graphicx}
\usepackage{braket}
\usepackage{enumerate}
\usepackage{multirow}
\usepackage{tikz}
\usetikzlibrary{quantikz}

\usepackage{caption}
\usepackage{subcaption}

\usepackage{algorithm,algorithmic}
\usepackage{balance}

\newtheorem{theorem}{Theorem}
\newtheorem{lemma}[theorem]{Lemma}
\newtheorem{corollary}{Corollary}[theorem]

\begin{document}

\title{Optimizing Ansatz Design in QAOA for Max-cut}

\author{Ritajit Majumdar}
\affiliation{Advanced Computing \& Microelectronics Unit, Indian Statistical Institute}
\email{majumdar.ritajit@gmail.com}
\author{Dhiraj Madan}
\affiliation{IBM Research, India}
\email{dmadan07@in.ibm.com}
\author{Debasmita Bhoumik}
\affiliation{Advanced Computing \& Microelectronics Unit, Indian Statistical Institute}
\author{Dhinakaran Vinayagamurthy}
\affiliation{IBM Research, India}
\author{Shesha S. Raghunathan}
\affiliation{IBM Systems, India}
\author{Susmita Sur-Kolay}
\affiliation{Advanced Computing \& Microelectronics Unit, Indian Statistical Institute}
\email{ssk@isical.ac.in}

\maketitle

\begin{abstract}
  Quantum Approximate Optimization Algorithm (QAOA) is studied primarily to find approximate solutions to combinatorial optimization problems. For a graph with $n$ vertices and $m$ edges, a depth $p$ QAOA for the Max-cut problem requires $2\cdot m \cdot p$ CNOT gates. CNOT is one of the primary sources of error in modern quantum computers. In this paper, we propose two hardware independent methods to reduce the number of CNOT gates in the circuit. First, we present a method based on Edge Coloring of the input graph that minimizes the the number of cycles (termed as depth of the circuit), and reduces upto $\lfloor \frac{n}{2} \rfloor$ CNOT gates. Next, we depict another method based on  Depth First Search (DFS) on the input graph that reduces $n-1$ CNOT gates, but increases depth of the circuit moderately. We analytically derive the condition for which the reduction in CNOT gates overshadows this increase in depth, and the error probability of the circuit is still lowered. We show that all IBM Quantum Hardware satisfy this condition. We simulate these two methods for graphs of various sparsity with the \textit{ibmq\_manhattan} noise model, and show that the DFS based method outperforms the edge coloring based method, which in turn, outperforms the traditional QAOA circuit in terms of reduction in the number of CNOT gates, and hence the probability of error of the circuit.
\end{abstract}


\section{Introduction}
\label{sec:intro}

Near term quantum devices have a small number of noisy qubits that can support execution of shallow depth circuits (i.e., those with few operational cycles) only. Variational Quantum Algorithms (VQA) aim to leverage the power as well as the limitations imposed by these devices to solve problems of interest such as combinatorial optimization \cite{farhi2014quantum,wang2018quantum,hadfield2019quantum,cook2020quantum}, quantum chemistry \cite{mcclean2016theory, grimsley2019adaptive}, and quantum machine learning \cite{10.1007/978-3-030-50433-5_45,biamonte2017quantum,torlai2020machine}. VQA divides the entire computation into functional modules, and outsources some of these modules to classical computers. The general framework of VQA can be divided into four steps: (i) encode the problem into a parameterized quantum state $\ket{\psi(\theta)}$ (called the ansatz), where $\theta = \{\theta_1,\theta_2, \hdots, \theta_k\}$ are $k$ parameters; (ii) prepare and measure the ansatz in a quantum computer, and determine the value of some objective function $C(\theta)$ (which depends on the problem at hand) from the measurement outcome; (iii) in a classical computer, optimize the set of parameters to find a better set $\theta' = \{\theta'_1,\theta'_2, \hdots, \theta'_k\}$ such that it minimizes (or maximizes) the objective function; (iv) repeat steps (ii) and (iii) with the new set of parameters until convergence.

Quantum Approximate Optimization Algorithm (QAOA) is a type of VQA that focuses on finding good approximate solutions to combinatorial optimization problems. It has been studied most widely for finding the maximum cut of a (weighted or unweighted) graph (called the Max-Cut problem) \cite{farhi2014quantum}. For this problem, given a graph $G = (V,E)$ where $V$ is the set of vertices and $E$ is the set of edges, the objective is to partition $V = V_1 \cup V_2$, such that $V_1 \cap V_2 = \phi$, and the number of edges crossing the partition is maximized. Throughout this paper, we shall consider {\it connected graphs} with $|V| = n$ and $|E| = m$, but the results can be easily extended to disconnected graphs as well.

In the initial algorithm proposed by Farhi \cite{farhi2014quantum} for the Max-Cut problem, a depth-$p$ QAOA consists of $p \geq 1$ layers of alternating operators on the initial state $\ket{\psi_0}$ 
\begin{equation}
\label{eq:ansatz}
    \ket{\psi(\gamma,\beta)} = ( \displaystyle \Pi_{l = 1}^{p} e^{(-i\beta_l H_M)} e^{(-i\gamma_l H_P)}) \ket{\psi_0}
\end{equation}

where $H_P$ and $H_M$ are called the Problem and Mixer Hamiltonian respectively, and $\gamma = \{\gamma_1, \gamma_2, \hdots, \gamma_p\}$ and $\beta = \{\beta_1, \beta_2, \hdots, \beta_p\}$ are the parameters. It is to be noted that the depth $p$ of the QAOA is not related to the depth of the quantum circuit realizing the algorithm. The problem Hamiltonian describing the Max-Cut can be represented as in Eq.~(\ref{eq:max_cut}), where $w_{jk}$ is the weight associated with the edge $(j,k)$.
\begin{equation}
\label{eq:max_cut}
    H_P = \frac{1}{2}\displaystyle \sum_{(j,k) \in E} w_{jk} (I - Z_j Z_k)
\end{equation}

Furthermore, the mixer Hamiltonian should be an operator that does not commute with the Problem Hamiltonian. In the traditional QAOA, the mixer Hamiltonian is $H_M = \displaystyle \sum_{i} X_i$.

Variations to this have been studied to improve the performance of the algorithm --- such as using other mixers \cite{bartschi2020grover,zhu2020adaptive, yu2021quantum}, training the parameters to reduce the classical optimization cost \cite{larkin2020evaluation}, and modifying the cost function for faster convergence \cite{barkoutsos2020improving}. In this paper we stick to the original problem and mixer hamiltonians proposed in the algorithm by Farhi et al. \cite{farhi2014quantum}. The applicability and effectiveness of our proposed method on the modifications of this algorithm can be looked at as a follow-up work. However, our proposed methods optimize the circuit corresponding to the problem hamiltonian. Since most of the modifications suggested in the literature aim to design more efficient mixers, our proposed optimization should be applicable on those as well.

The realization of the QAOA circuit for Max-cut requires two CNOT gates for each edge (details given in Sec.~\ref{sec:ansatz}). Hardware realization of a CNOT gate is, in general, significantly more erroneous than a single qubit gate. Even in the higher end devices of IBM Quantum, such as \textit{ibmq\_montreal}, \textit{ibmq\_manhattan}, the probability of error for a single qubit gate and a CNOT gate are $\mathcal{O}(10^{-4})$ and $\mathcal{O}(10^{-2})$, respectively~\cite{ibmquantum}. In other words, CNOT gates are $100$ times more likely to be erroneous than single qubit gates. Therefore, we focus primarily on reducing the number of CNOT gates in the design of QAOA ansatz for Max-cut.

\subsubsection*{Contributions of this paper}

In this paper, we

\begin{enumerate}[(i)]
    \item propose two optimization methods for reducing the number of CNOT gates in the first layer of the QAOA ansatz based on  (1) an Edge Coloring that can reduce upto $\lfloor \frac{n}{2} \rfloor$ CNOT gates, and (2) a Depth First Search (DFS) that can reduce $n-1$ CNOT gates.
    
    \item prove that there exists no method that can reduce more than $n-1$ CNOT gates while still maintaining a fidelity of 1 with the original QAOA ansatz \cite{farhi2014quantum}.
    
    \item show that while the Edge Coloring based optimization reduces the depth of the circuit, the DFS based method may increase the depth. We further analytically derive the criteria (involving the increase in the depth and the reduction in the number of CNOT gates) for which the DFS based optimization method still leads to a lower probability of error in the circuit, and show that the IBM Quantum Hardwares \cite{ibmquantum} conform to that criteria.
    
    
    \item simulate our proposed optimization methods in Qiskit~\cite{Qiskit} with the \textit{ibmq\_manhattan} noise model and show that for graphs of different sparsity (Erdos-Renyi graphs with the probability of edge varying from 0.4 - 1)
    \begin{enumerate}
        \item the proposed reduction in the CNOT gate is still retained post transpilation
        \item the DFS based method has lower error probability than the Edge Coloring method, which in its turn has lower error probability than the traditional QAOA ansatz.
    \end{enumerate}
\end{enumerate}

Therefore, for any graph $G = (V,E)$, our proposed method provides reduction in the number of CNOT gates, and hence lowers the error probability of the circuit. Although the DFS method provably surpasses the Edge Coloring method, both in terms of reduction in CNOT gates and lowering the error probability, the latter reduces the depth of the QAOA circuit, and is also used in the DFS based method to arrange the edges which do not form a part of the DFS tree.

For the rest of this paper, we consider \emph{unweighted and connected graphs}, $i.e.$, $w_{jk} = 1$, $\forall$ $(j,k) \in E$. However, the circuit corresponding to the ansatz does not change if we have a weighted graph \cite{hadfield2019quantum}. Therefore, every analysis in this paper holds for a weighted graph as well. Furthermore, the analysis of this paper will hold as it is, or with some minimal modification, for disconnected graphs as well.

The rest of the paper is organized as follows - Section~\ref{sec:ansatz} briefly discusses the traditional QAOA ansatz design. In Section~\ref{sec:thm} we provide the proposed optimization and the criteria for it. Section~\ref{sec:edge_col} and ~\ref{sec:dfs} describe two methods of optimization based on Edge Coloring and DFS respectively. We provide the respective algorithms and analyze the conditions under which each one reduces the probability of error. We present the results of our simulation in section~\ref{sec:sim} and conclude in Section~\ref{sec:con}.

\section{Traditional ansatz design for QAOA}
\label{sec:ansatz}

The objective function of a depth-$p$ QAOA for Max-Cut~\cite{farhi2014quantum} can be expressed as
\begin{equation}
\label{eq:objective}
    \max_{\psi(\gamma,\beta)} \bra{\psi(\gamma,\beta)} H_P \ket{\psi(\gamma,\beta)}
\end{equation}
where $\gamma = \{\gamma_1, \gamma_2, \hdots, \gamma_p\}$ and $\beta = \{\beta_1, \beta_2, \hdots, \beta_p\}$ are the parameters. The trial wavefunction $\ket{\psi(\gamma,\beta)}$ is called the ansatz. The QAOA ansatz has a fixed form as described in Eq.~(\ref{eq:ansatz}). The initial state $\ket{\psi_0}$ is usually the equal superposition of $n$ qubits, where $n = |V|$. Note that the depth of the circuit required to prepare $\ket{\psi_0}$ is 1 (Hadamard gates acting simultaneously on all the qubits). Similarly, for each layer of QAOA, the operator $exp(-i \beta_l H_M)$ can be realized by a depth one circuit of $R_x(\beta_l)$ gates acting simultaneously on all the qubits.

The operator $exp(-i \gamma_l H_P)$ has a more costly implementation. Note that
\begin{eqnarray*}
exp(-i \gamma_l H_P) = \displaystyle \Pi_{(i,j) \in E} exp\left(-i \gamma_l \left(\frac{I-Z_j Z_k}{2}\right)\right).
\end{eqnarray*}

The operator $exp\left(-i \gamma_l \left(\frac{I-Z_j Z_k}{2}\right)\right)$ acts on each edge $(j,k)$, and is realized as shown below:
%
\begin{figure}[H]
\centering
	\begin{quantikz}
		{q_{j}}&&\ctrl{1} & \qw & \ctrl{1} & \qw \\
		{q_{k}}&&\targ{} & \gate{R_z(2\gamma_l)} & \targ{} & \qw
	\end{quantikz}
	\label{fig:z_jz_k}
\end{figure}

Here, $q_j$ and $q_k$ represent qubit indices $j$ and $k$, respectively. Note that Max-Cut is a symmetric problem, and therefore, the selection of control and target from qubits $q_j$ and $q_k$ for the CNOT gate corresponding to the edge $(j,k)$ is irrelevant, $i.e.$ the operator $exp\left(-i \gamma_l \left(\frac{I-Z_j Z_k}{2}\right)\right)$ can be equivalently realized as $CNOT_{kj} (I_k \otimes R_z(2\gamma_l)_{j})$ $CNOT_{kj}$. In Fig.~\ref{fig:depth_qaoa}(a) and (b), we show a 2-regular graph with four vertices and its corresponding QAOA circuit for $p = 1$ respectively.

\begin{figure}[htb]
     \centering
     \begin{subfigure}[b]{0.3\textwidth}
         \centering
         \includegraphics[scale=0.35]{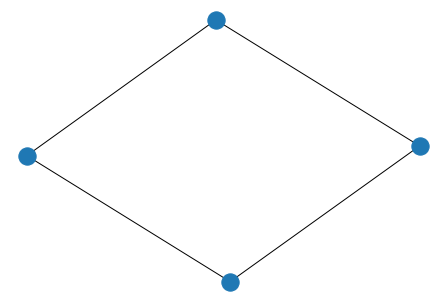}
         \caption{A 2-regular graph with four vertices}
         \label{col}
     \end{subfigure}
     \hfill
     \begin{subfigure}[b]{0.4\textwidth}
         \centering
         \includegraphics[scale=0.4]{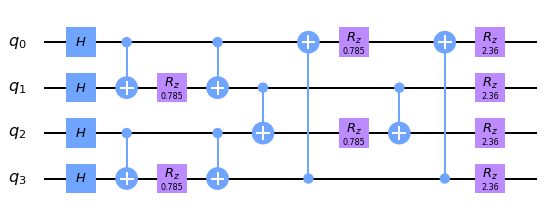}
         \caption{Max-Cut QAOA circuit for $p=1$  corresponding to the graph}
         \label{dfs}
     \end{subfigure}
    \caption{The Max-Cut QAOA circuit for $p=1$  corresponding to the 2-regular graph with four vertices; the values of $\gamma$ and $\beta$ can be arbitrary but those in this figure are the optimum values for this graph when $p = 1$}
    \label{fig:depth_qaoa}
\end{figure}

\section{Methods for Optimized ansatz design}
\label{sec:thm}

Some recent studies have proposed optimization methods for the circuit of the QAOA ansatz with respect to the underlying hardware architecture \cite{alam2020circuit}. In this paper we propose two \textit{hardware independent} methods to reduce the number of CNOT gates in the traditional QAOA ansatz. The intuition is that in the circuit realization of the operator $exp\left(-i \gamma_l \left(\frac{I-Z_j Z_k}{2}\right)\right)$ as $CNOT_{jk} (I_j \otimes R_z(2\gamma_l)_{k})$ $CNOT_{jk}$, the first CNOT gate can be removed whenever it does not make any contribution to the overall effect of the operator. Our proposed method reduces the number of CNOT gates in the circuit irrespective of the hardware architecture, and hence is applicable for any quantum device.

In Theorem~\ref{thm:equiv} we prescribe the condition where the first CNOT gate is irrelevant to the effect of the said operator, and hence may be removed.

\begin{theorem}
\label{thm:equiv}
Let $\ket{\psi}$ be an $n$-qubit state prepared in a uniform superposition (upto relative phase) over all basis states $\ket{x_1, \hdots, x_n}$ such that the relative phase on each basis state is a function of a subset $S \subset$ $\{1,2,...,n\}$ of the $n$ qubits (and independent of remaining qubits) $i.e.$
\begin{center}
    $\ket{\psi} = \frac{1}{\sqrt{2^n}}\displaystyle \sum_{x_1,...,x_n} e^{(i \phi(x_S))}\ket{x_1,...,x_n}$
\end{center}
where $x_S = \{x_i : i \in S\}$ and $\phi(x_S)$ depicts the relative phase of each superposition state. For any two qubits $\ket{j}$ and $\ket{k}$, where $\ket{k} \notin S$, and for the two operators $U_1 = CNOT_{jk} (I_j \otimes R_z(2\gamma_l)_{k}) CNOT_{jk}$ and $U_2 = (I_j \otimes R_z(2\gamma_l)_{k}) CNOT_{jk}$, we have 
\begin{center}
    $U_1\ket{\psi} = U_2\ket{\psi}$.
\end{center}
\end{theorem}

\begin{proof}
Let us consider the action of the operators $U_1$ and $U_2$ on any edge $(j,k)$.
\begin{equation}
U_1\ket{\psi} = CNOT_{jk} (I_j \otimes R_z(2\gamma_l)_{k}) (CNOT_{jk}) \ket{\psi} \nonumber
\end{equation}
\begin{align}
=&\displaystyle \sum_{x_1,...,x_n} CNOT_{jk} (I_j \otimes R_z(2\gamma_l)_{k}) (CNOT_{jk})\nonumber \\
& e^{i\phi(x_S)} \ket{x_1,...,x_n}  \\
=&\displaystyle \sum_{x_1,...,x_n} CNOT_{jk} (I_j \otimes R_z(2\gamma_l)_{k}) \nonumber \\
& e^{i \phi(x_S)} \ket{x_1,..,x_k'=x_j \oplus x_k,.,x_n}  \\
=&\displaystyle  \sum_{x_1,...,x_n} e^{i(\phi(x_S)- \gamma_l (x_j \oplus x_k))} CNOT_{jk}  \nonumber \\
&  \ket{x_1,..,x_k'=x_j \oplus x_k,.,x_n}  \\
\label{eq:u1}
=& \displaystyle \sum_{x_1,...,x_n} e^{i(\phi(x_S) - \gamma_l (x_j \oplus x_k))}\ket{x_1,...,x_n}
\end{align}
where $e^{i \phi(x_S)}$ is the cumulative effect of operators acting on previous edges (= 0 if $(j,k)$ is the first). We have dropped the normalization constant for brevity.

Similarly,
\begin{equation}
    U_2\ket{\psi} = CNOT_{jk} (I_j \otimes R_z(2\gamma_l)_{x_k}) \ket{\psi} \nonumber
\end{equation}
\begin{equation}
    = CNOT_{jk} \displaystyle \sum_{x_1,...,x_n} e^{i((\phi(x_S)) - \gamma_l x_k)}\ket{x_1,...,x_n} \nonumber
\end{equation}
\begin{equation}
\label{eq:u2_mid}
    = \displaystyle \sum_{x_1,...,x_n} e^{i((\phi(x_S)) - \gamma_l x_k)} \ket{x_1,..,x_j \oplus x_k,..,x_n}
\end{equation}

where the qubit in $k^{\text{th}}$ position changes to $x_j \oplus x_k$ due to the $CNOT_{jk}$ operation. Now, substituting $x_k' = x_j \oplus x_k$ in the above equation, we get

\begin{equation}
    U_2\ket{\psi} = \displaystyle \sum_{x_1,...,x_n} e^{i((\phi(x_S)) - \gamma_l x_k)} \ket{x_1,..,x_j \oplus x_k,..,x_n} \nonumber
\end{equation}
\begin{equation}
    = \displaystyle \sum_{x_1,..,x_k',..,x_n} e^{i((\phi(x_S)) - \gamma_l (x_j \oplus x_k'))} \ket{x_1,..,x_k',..,x_n} \nonumber
\end{equation}
\begin{equation}
\label{eq:u2}
    = \displaystyle \sum_{x_1,..,x_k,..,x_n} e^{i((\phi(x_S)) - \gamma_l (x_j \oplus x_k))} \ket{x_1,..,x_k,..,x_n}
\end{equation}

Here since $k \notin S$, the substitution in second last step, does not change the phase $e^{i \phi(x_S)}$. The last step is valid since $x_k'$ is a running index and hence can be changed to $x_k$. Thus Eq.~(\ref{eq:u1}) and Eq.~(\ref{eq:u2}) are identical.
\end{proof}

\begin{corollary}
\label{cor:cond}
For a graph $G$, we can optimize the circuit for the operator $exp\left(-i \gamma_l \left(\frac{I-Z_j Z_k}{2}\right)\right)$ corresponding to an edge $(j,k)$ replacing $U_1$ by $U_2$,  provided that the target vertex does not occur in any of the edge operators run earlier. In other words, the following conditions are sufficient to optimize an edge:-
\begin{enumerate}
    \item if the vertex $j$ is being operated on for the first time, then it acts either as a control or a target for the CNOT gate corresponding to the operator;
    \item the vertex $j$ does not act as a target of the CNOT gate if it occurs as a part of any other edge operators run earlier.
\end{enumerate}
\end{corollary}

\begin{proof}

The first time we consider an edge adjacent to a vertex $j$, where $j \notin x_S$, (see Theorem~\ref{thm:equiv}) the relative phase $\phi(x_S)$ does not depend on $j$. Thus it satisfies the condition of Theorem~\ref{thm:equiv} and allows optimization of the operator.

 On the other hand, if the vertex $j$ occurs as part of an edge operator already run, the phase on the basis state $\phi$ can potentially depend on $S$, $i.e.$ $j \in S$. By not allowing it to act as target, we satisfy the conditions of Theorem~\ref{thm:equiv}.
\end{proof}

From the above discussion, it follows that if we arbitrarily choose edges for applying the operator $exp\left(-i \gamma_l \left(\frac{I-Z_j Z_k}{2}\right)\right)$, then it cannot be guaranteed that a large number of edges will conform to  Corollary~\ref{cor:cond}. The requirement, in fact, imposes a precedence ordering among the edges. In Section~\ref{sec:edge_col} and ~\ref{sec:dfs}, we provide two algorithmic procedures for maximizing the number of edges that satisfy the requirement in order to reduce the number of CNOT gates in the ansatz.

For the rest of the paper, we say that {\it an edge is optimized} if the operator $U_2$ can be operated on that edge instead of $U_1$.

\section{Edge Coloring based Ansatz Optimization }
\label{sec:edge_col}
The total error one incurs in a circuit depends on the number of operators (since a larger number of operators tend to incur more error) and the depth of the circuit (corresponding to relaxation error). In this section, we discuss how one can minimize the depth of the circuit. We also discuss the possibility of reduction in CNOT gates in the depth optimized circuit. 

The operators $H_M$ act on distinct qubits and hence can be run in parallel contributing to a depth of 1 (for each step of the QAOA). On the other hand, the operators in $H_P$ can potentially contribute a lot to depth since the edge operators do not act on disjoint vertices.
At a given level of the circuit, we can only apply edge operators corresponding to a vertex disjoint set of edges. Thus the minimum depth of the circuit will correspond to the minimum value $k$ such that we can partition the set of edges $E$ as a disjoint union $\cup_i E_i$ where each subset $E_i$ consists of a vertex disjoint collection of edges. This in turn corresponds to the edge coloring problem in a graph.

Given a graph $G = (V,E)$ and a set of colors $\chi' = \{\chi'_1, \chi'_2, \hdots, \chi'_k\}$, an edge coloring \cite{west2001introduction} assigns a color to each edge $e \in E$, such that any two adjacent edges ($i.e.$, edges incident on a common vertex) must be assigned distinct colors. The edge coloring problem comprises of coloring the edges using the minimum number of colors $k$. Note that the operators corresponding to edges having the same color can therefore be executed in parallel.
Moreover,
\begin{enumerate}
    \item the number of colors in optimal coloring, called the chromatic index, corresponds to the minimum depth of the circuit;
    \item edges having the same color corresponds to the operators $exp\left(-i \gamma_l \left(\frac{I-Z_j Z_k}{2}\right)\right)$ that can be executed simultaneously.
\end{enumerate}

Optimal edge coloring is an NP-complete problem \cite{west2001introduction}.  But it is not practical to allocate exponential time to find the optimal edge-coloring as a pre-processing step for QAOA. Vizing's Theorem states that every simple undirected graph can be edge-colored using at most $\Delta + 1$ colors, where $\Delta$ is the maximum degree of the graph \cite{vizing1964estimate}. This is within an additive factor of 1 since any edge-coloring must use at least $\Delta$ colors. Misra and Gries algorithm \cite{misra1992constructive} achieves the above bound constructively in $\mathcal{O}(n\cdot m)$ time. Therefore, we use the Misra and Gries edge coloring algorithm.  Algorithm~\ref{alg:edgecol} below computes the sets of edges having the same color using  Misra and Gries algorithm as a subroutine. It returns the largest set $S_{max}$ of edges having the same color in the coloring computed by Misra and Gries algorithm.

\begin{algorithm}[H]
\caption{Edge Coloring based Ansatz Optimization}
\label{alg:edgecol}
\begin{algorithmic}[1]
\REQUIRE A graph $G = (V,E)$.
\ENSURE Largest set $S_{max}$ of edges having the same color.
\STATE Use the Misra and Gries algorithm to color the edges of the graph $G$.
\STATE $S_i \leftarrow$ set of edges having the same color $i$, $1 \leq i \leq \chi'$.
\STATE $S_{max} \leftarrow$ $max\{S_1, S_2, \hdots, S_{\chi'}\}$.
\STATE Return $S_{max}$.
\end{algorithmic}
\end{algorithm}

This edge coloring approach provides the minimum depth achievable for QAOA ansatz using a polynomial time pre-processing. After reducing the depth, we now try to further reduce errors by decreasing the number of CNOT gates.
Recall that the operators corresponding to edges with the same color can be executed in parallel. We use the operators corresponding to the edges of $S_{max}$ as the first layer of operators. The other layers can be used in any order.

\begin{lemma}
Every edge in the first layer can be optimized according to Corollary~\ref{cor:cond}.
\end{lemma}

\begin{proof}
For every edge $(u,v)$ in the first layer, both the vertices are adjacent to an edge for the first time, $i.e.$, both $u, v \notin S$. Therefore, it satisfies the criteria of Corollary~\ref{cor:cond}, and hence can be optimized. In fact, any one of the qubits corresponding to the two vertices can be selected as the control for the CNOT operation.
\end{proof}

Some edges in the corresponding layers may be optimized as well. Nevertheless, it is trivial to come up with examples where this is not the case (e.g., a complete graph of 4-vertices). Therefore, in the worst case scenario, only the edges in the first layer can be optimized. However, since this method does not increase the depth of the circuit, it always leads to a more efficient circuit design than the traditional QAOA circuit with lower depth (by 1 since the first layer of CNOT is absent) and fewer CNOT gates.

For general graphs, the worst case scenario is, therefore, that only the edges in the first layer can be optimized. In the following subsection we provide an analysis on the number of optimized edges using this method.

\subsection{Lower and upper bound on the number of optimized edges}
Let us assume that the chromatic index of a graph $G = (V,E)$ is $\chi'$. Using the Misra and Gries Theorem \cite{misra1992constructive} we can find a polynomial time coloring using at most $\Delta + 1$ colors, where $\Delta$ is the maximum degree of the graph. Therefore, on an average, $\lceil \frac{m}{\Delta + 1} \rceil$ edges have the same color.

More precisely, two extreme cases arise: (i) the colors may be uniformly distributed, and the maximum number of edges having the same color is $\lceil \frac{m}{\Delta + 1} \rceil$; or (ii) one of the colors is used dominantly for most of the edges. Nevertheless, note that for all the edges adjacent to the same vertex, a particular color can be assigned to one of the edges only. Therefore, the dominant color can be used at most on $\lfloor \frac{n}{2} \rfloor$ edges, where $n = |V|$. Therefore, the possible number of optimized edges that can be obtained via the Edge Coloring method is as shown in Eq.~(\ref{eq:edge_col}).
\begin{equation}
\label{eq:edge_col}
    \lceil \frac{m}{\Delta + 1} \rceil \leq ~\# ~Optimized ~Edges \leq \lfloor \frac{n}{2} \rfloor.
\end{equation}

\section{Depth First Search based Ansatz Optimization}
\label{sec:dfs}
As the edge coloring based algorithm can optimize at most $\lfloor \frac{n}{2} \rfloor$ edges, in this section, we present a Depth First Search (DFS) based optimization procedure which can optimize $n-1$ edges. Algorithm~\ref{alg:dfs}, for obtaining the optimized QAOA ansatz, uses the standard DFS algorithm \cite{cormen2009introduction}, by returning the tree edges or discovery edges forming the DFS tree.

In this method, we start from the first vertex of the DFS tree. For every edge $e = (u,v)$ in the DFS tree, the vertex $u$ is made the control and $v$ is made the target for the CNOT gate corresponding to that edge. The edges are operated on sequentially one after another, as in the set $E_{dfs}$ (the tree edges). Once every edge in the DFS tree has been operated on, the remaining edges can be executed in any order. In fact, one may opt to use the Edge Coloring method on the remaining edges to obtain the minimum depth of the circuit corresponding to these edges, although CNOT gates cannot be reduced any further.

\begin{algorithm}[H]
\caption{DFS based Ansatz Optimization }
\label{alg:dfs}
\begin{algorithmic}[1]
\REQUIRE A graph $G = (V,E)$.
\ENSURE A list $E_{dfs}$ of $n-1$ edges.
\STATE $E_{dfs} = \{\}$
\STATE $u \leftarrow$ randomly selected vertex from $V$.
\STATE Start DFS from the vertex $u$. For every vertex $v$ discovered from its predecessor $v'$, $E_{dfs} = E_{dfs} \cup (v',v)$.
\STATE Return $E_{dfs}$.
\end{algorithmic}
\end{algorithm}

\begin{theorem}
Each edge in the DFS tree can be optimized according to Corollary~\ref{cor:cond}.
\label{thm:dfs}
\end{theorem}

\begin{proof}
We prove this by the method of induction. Let $u$ be the vertex from which the DFS tree starts. Then $u$ is being operated on for the first time, and, hence, can act both as a control/target for the CNOT operation corresponding to the first edge (Corollary~\ref{cor:cond}). Choose $u$ to be the control.

\textbf{Base case}: If $v$ is the vertex that is discovered from $u$ via the edge $(u,v)$, then choosing $u$ as the control and $v$ as the target satisfies Corollary~\ref{cor:cond}. Therefore, the edge $(u,v)$ can be optimized.

\textbf{Induction hypothesis}: Let the DFS tree has been constructed upto some vertex $j$, and every edge $(e_1, e_2)$ in this DFS tree so far can be optimized, $i.e.$ $e_1$ acts as the control and $e_2$ as the target.

\textbf{Induction step}: Let the next vertex in the DFS tree, that is discovered from some vertex $i$, is $k$. From DFS algorithm, the vertex $i$ must have been discovered in some previous step. Since vertex $k$ was not previously discovered, so $k \notin x_S$ and hence the edge $(i,k)$ can be optimized if we select $i$ to be the control and $k$ as the target.
\end{proof}

Therefore, the DFS based optimization method provides $n-1$ optimized edges, $i.e.$, a reduction in the number of CNOT gates by $n-1$. We now show in Theorem~\ref{thm:optimal} that this is the maximum number of edges that can be optimized.

\begin{figure}[H]
     \centering
     \begin{subfigure}[b]{0.45\textwidth}
         \centering
         \includegraphics[scale=0.4]{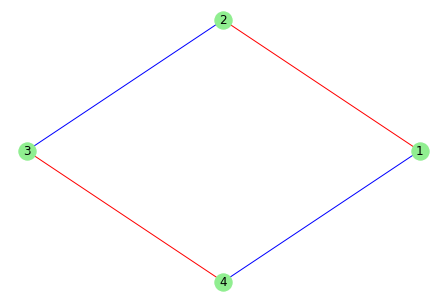}
         \caption{Edge Coloring Based Optimization}
         \label{col}
     \end{subfigure}
     \hfill
     \begin{subfigure}[b]{0.45\textwidth}
         \centering
         \includegraphics[scale=0.4]{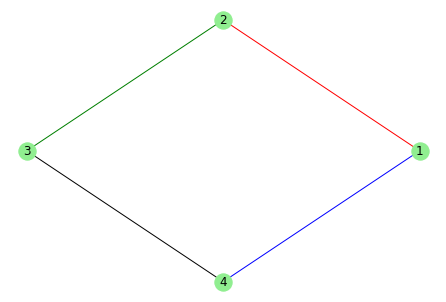}
         \caption{Depth First Search Based Optimization}
         \label{dfs}
     \end{subfigure}
    \caption{Depth of the ansatz circuit when using (a) Edge Coloring and (b) DFS based method; edges having same color can be executed simultaneously. The depth of the spanning tree in the DFS based method is 4, compared to depth 2 for the Edge Coloring based method. However, the number of optimized edges in the Edge Coloring based method is 2, while that by the DFS based method is 3.}
    \label{fig:depth}
\end{figure}

\begin{theorem}
Optimization of ansatz for Max-Cut QAOA with p=1, by deletion of the CNOT gate in the first level for an edge of the graph, can be done for no more than $n-1$ edges.
\label{thm:optimal}
\end{theorem}

\begin{proof}
Let us assume that there is some method by which at least $n$ edges can be optimized. Now, the connected subgraph which contains all the $n$ vertices and at least $n$ optimized edges must contain a cycle. Let $(u,v)$ be an edge of this cycle, $i.e.$, if $(u,v)$ is removed then the residual graph is a tree (in case there are $> n$ edges, the removal of edges can be performed recursively till such an edge $(u,v)$ is obtained whose removal makes the residual graph a tree). For this edge $(u,v)$, both the vertices $u$ and $v$ are endpoints of some other optimized edges as well. Therefore, from Corollary~\ref{cor:cond} both $u$ and $v$ must act as the control for the CNOT gate corresponding to the edge $(u,v)$ in order for this edge to be optimized, which is not possible. Therefore, it is not possible to optimize more than $n-1$ edges.
\end{proof}

Therefore, the DFS method is optimal in the number of optimized edges. However, we note that the DFS based method associates an ordering of the edges, $i.e.$, some of the edges which could have been operated on simultaneously, cannot be done so now.
This, in turn, can lead to an increase in the depth of the circuit. Hence, a penalty for this method producing optimal reduction in CNOT gates, is that it increases the depth of the circuit. 

\begin{figure}[htb]
     \centering
     \begin{subfigure}[b]{0.45\textwidth}
         \centering
         \includegraphics[scale=0.37]{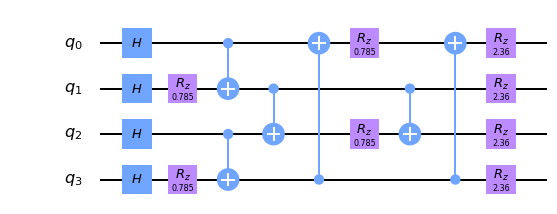}
         \caption{Optimized circuit by Edge Coloring based method }
         \label{col}
     \end{subfigure}
     \hfill
     \begin{subfigure}[b]{0.48\textwidth}
         \centering
         \includegraphics[scale=0.37]{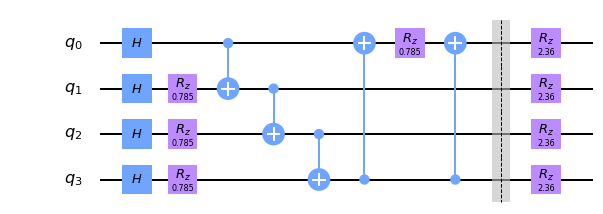}
         \caption{Optimized circuit by DFS based method}
         \label{dfs}
     \end{subfigure}
    \caption{Max-Cut QAOA ansatz with $p=1$ corresponding to (a) Edge Coloring and (b) DFS based optimization. In (a), the first CNOT gates of the operators have been deleted. The operators corresponding to $(q_1,q_2)$ and $(q_3,q_0)$ act in parallel. In (b), the first CNOT gates of three operators have been deleted, but the depth has increased.}
    \label{fig:opt}
\end{figure}

In Fig.~\ref{fig:depth}, we show a 2-regular graph with four vertices. In Fig.~\ref{fig:depth}(a), the depth of the circuit corresponding to the operator $exp(-i\gamma_l H_P)$ is 2; the edges of the same color can be operated on simultaneously. If the red (or blue) edges form the first layer, then those two edges are optimized. However, if we use the DFS method, with the DFS tree starting from, say, vertex 1, then the edges $(1,2),(2,3)$ and $(3,4)$ can be optimized (Fig.~\ref{fig:depth}(b)). Now these three edges must be operated on one after another, followed by the fourth edge. Thus the depth of the circuit corresponding to the operator $exp(-i\gamma_l H_P)$ becomes 4. The circuits corresponding to these two scenarios are depicted in Fig~\ref{fig:opt}(a) and (b) respectively.


The question, therefore, is whether this increase in depth is always acceptable, even with the increased reduction in the number of CNOT gates as, with increased depth, the circuit becomes more prone to relaxation error.
Numerical analysis and simulation (Section~\ref{sec:sim}) establises that although the depth of the circuit is increased, the overall error probability of the circuit is reduced further.

\subsection{When is the DFS based method useful?}

In this subsection, we formulate a relation for which the increase in the depth still leads to a lower probability of error for the reduction in the number of CNOT gates. For this analysis, we make an assumption that the error in the circuit arises only from noisy CNOT gates and the depth of the circuit ($i.e.$, the $T_1$ time). Although this assumption is idealistic, the ansatz primarily consists of layers of CNOT gates. Furthermore, in superconductors, $R_z$ gates are executed virtually \cite{mckay2017efficient}, and hence does not lead to any gate error. Therefore, CNOT is the primary source of gate error and with increasing depth, the qubits become more prone to relaxation error. Therefore, this assumption allows for a simple but powerful model for analyzing the query at hand.

Let us assume that the time duration and the error probability of each CNOT gate is $t_{cx}$ and $p_{cx}$ respectively. Let there be $N$ layers of CNOT operations. Note that although there can be multiple CNOT gates in each layer, the time duration of each layer is $t_{cx}$ only. Therefore, the probability of no error ($i.e.$, the probability that the circuit remains error free) after $N$ layers of operations, considering only relaxation error, is
    $exp(-\frac{N t_{cx}}{T_1})$.

Let there be $k$ CNOT gates in the original circuit. Therefore, the probability of no error after the operation of the CNOT gates, considering only CNOT gate error, is
    $(1 - p_{cx})^k$.

Combining both the sources of the errors, Eq.~(\ref{eq:no_err}) gives the probability of success ($i.e.$, the probability of no error) after a single cycle of computation of the QAOA ansatz.
\begin{equation}
    \label{eq:no_err}
    P_{success} = (1 - p_{cx})^k \cdot exp(-\frac{N t_{cx}}{T_1})
\end{equation}

Henceforth, $P_{success}$ will refer to the  probability of success ($i.e.$, how close the noisy outcome is to the noise-free ideal outcome) of the ansatz circuit execution for a single run of the algorithm. Note that in QAOA, the ansatz is computed multiple times for multiple cycles, and the objective is to maximize the expectation value of the outcome.

We further assume that after the optimization using DFS based method, $k_1$ CNOT gates have been reduced leading to an increase in $N_1$ layers of operations. The probability that this optimized circuit remains error-free is given in Eq.~(\ref{eq:opt_err}).
\begin{equation}
    \label{eq:opt_err}
    P^{opt}_{success} = (1 - p_{cx})^{(k-k_1)} \cdot exp(-\frac{(N+N_1) t_{cx}}{T_1})
\end{equation}

The optimization is fruitful only when $P^{opt}_{success} \geq P_{success}$. Note that
\begin{center}
    $P^{opt}_{success} = P_{success} \cdot exp(-\frac{N_1 t_{cx}}{T_1}) / (1 - p_{cx})^{k_1}$
\end{center}

Since both $P^{opt}_{success}$ and $P_{success} \leq 1$, the required inequality holds only if $exp(- \frac{N_1 t_{cx}}{T_1}) / (1 - p_{cx})^{k_1} \geq 1$. In other words,
\begin{eqnarray}
\label{eq:cond1}
exp(-\frac{N_1 t_{cx}}{T_1}) &\geq& (1 - p_{cx})^{k_1} \nonumber\\
\Rightarrow N_1 &\leq& \lambda \times k_1, \nonumber\\
~where && \lambda = \frac{-ln(1 - p_{cx}) \cdot T_1}{t_{cx}}.
\end{eqnarray}

The constant $\lambda$ is defined in terms of parameters specific to the quantum device.

\subsubsection{Effect of varying $\lambda$}
Given that $\lambda = f(t_{cx},p_{cx},T_1)$,  we expect the $T_1$ value to increase, and the $t_{cx}$ and $p_{cx}$ values to decrease as technology improves. The value of $\lambda$ increases for increasing $T_1$ and/or decreasing $t_{cx}$, whereas it decreases for decreasing $p_{cx}$. Therefore, 

\begin{itemize}
    \item If $p_{cx}$, the probability of error for CNOT gates decreases, the optimization becomes less useful since we are increasing the probability of relaxation error, but the reduction in error probability becomes less. As per this observation, for smaller  $\lambda$, Eq.~(\ref{eq:cond1}) is satisfied when the increase in depth is reduced as well.
    
    \item Similarly, if (i) $T_1$ increases, then the qubit can retain its conherence for a longer period of time, or (ii)  $t_{cx}$ decreases, then the overall computation time of the circuit decreases as well, and the circuit can allow some relaxation in the depth even if $T_1$ remains unchanged. We observe that for both of these cases, by Eq.~(\ref{eq:cond1}),  $\lambda$ increases, thus allowing more increase in depth for a given reduction in the number of CNOT gates.
\end{itemize}

\subsection{Trade-off between depth and reduction in CNOT gates}

If the DFS based method is not applied, then the number of layers of CNOT gates is equal to the number of color classes (as in Edge Coloring method). The maximum number of color classes is $\Delta + 1$ (as discussed in the previous section), and hence the maximum depth of the circuit is $\Delta + 1$ as well. Now, when the DFS based method is applied, the circuit can be divided into two disjoint sets of edges:
\begin{enumerate}
    \item The set of edges belonging to the DFS tree which can be optimized. The depth of this portion of the circuit is at most $n-1$ ($i.e.$, the depth of the DFS tree). Each of the operators corresponding to these edges contains a single CNOT gate only, and hence the number of CNOT gate layers is $n-1$ as well.
    \item The set of edges that do not belong to the DFS tree and hence are not optimized. The operators corresponding to these edges can be applied in any order, but after all the optimized edges. When removing the edges of the DFS tree, the degree of each vertex is reduced by at least 1. Therefore, the maximum degree of the remaining subgraph is at most $\Delta - 1$. Therefore, the depth of this portion of the circuit will be at most $\Delta$ (From Misra and Gries Algorithm). Each of the layer in this portion contains 2 CNOT gates, and hence the number of CNOT gate layers is $2\Delta$.
\end{enumerate}

Therefore, the maximum depth of the circuit after applying the DFS based optimization is $n-1 + \Delta$. In other words, the increase in depth due to this method is given by Eq.~(\ref{eq:cond2}).
\begin{eqnarray}
\label{eq:cond2}
    n-1 + \Delta - (\Delta + 1) = n - 2
\end{eqnarray}

Recall that the number of CNOT gates reduced due to the DFS method is always $n-1$. Therefore, from Eq.~(\ref{eq:cond1}) and $~(\ref{eq:cond2})$, we get
\begin{eqnarray}
\label{eq:condition}
n - 2 &\leq& \lambda \cdot (n-1) \nonumber \\
\Rightarrow \lambda &\geq& \frac{n-2}{n-1}
\end{eqnarray}

In Table~\ref{tab:lambda} we show the average value of $\lambda$ for some IBM Quantum~\cite{ibmquantum} devices, ranging from the comparatively more noisy \textit{ibmq\_melbourne} to the comparatively less noisy \textit{ibmq\_manhattan}.

\begin{table}[htb]
    \centering
    \caption{Average value of $\lambda$ for four IBM Quantum machines ~\cite{ibmquantum}}
    \begin{tabular}{|c|c|}
    \hline
    IBM Quantum devices & Avg value of $\lambda$\\
    \hline
    \textit{ibmq\_manhattan} & 3.6\\
    \hline
    \textit{ibmq\_montreal} & 2.47\\
    \hline
    \textit{ibmq\_sydney} & 3.35\\
    \hline
    \textit{ibmq\_melbourne} & 2.03\\
    \hline
    \end{tabular}
    \label{tab:lambda}
\end{table}

Note that the lower bound on $\lambda$, $\frac{n-2}{n-1}$ (Eq.~(\ref{eq:condition}), is always less than $1$ for all $n$. In the asymptotic limit, $\frac{n-2}{n-1} \rightarrow 1$. Thus, the proposed DFS based optimization method leads to a lower error probability on any quantum device for which $\lambda \geq 1$. Table~\ref{tab:lambda} readily shows that the IBM Quantum hardwares conform to this requirement.


\section{Results of simulation}
\label{sec:sim}

\begin{table*}[t]
    \centering
    \caption{Comparison of Max-Cut QAOA ansatz circuits post transpilation on \textit{ibmq\_manhattan}: (i) Traditional, (ii) Edge coloring and (iii) DFS based optimization}
    \begin{tabular}{|c|c|c|c|c|}
    \hline
    \multirow{2}{*}{Graph Family} & \multirow{2}{*}{\# qubits} & \multicolumn{3}{c|}{\# CNOT gates in Max-Cut QAOA ansatz circuit}\\
    \cline{3-5}
    & & Traditional  & Edge coloring & DFS\\
    \hline
    \multirow{6}{*}{Complete graph} & 10 & 90 & 85 & 81\\
    \cline{2-5}
    & 20 & 380 & 370 & 361\\
    \cline{2-5}
    & 30 & 870 & 855 & 841\\
    \cline{2-5}
    & 40 & 1560 & 1540 & 1521\\
    \cline{2-5}
    & 50 & 2450 & 2425 & 2401\\
    \cline{2-5}
    & 60 & 3540 & 3510 & 3481\\
    \hline
    \multirow{6}{*}{Erdos-Renyi ($p_{edge}$ = 0.8)} & 10 & 70 & 66 & 61\\
    \cline{2-5}
    & 20 & 302 & 292 & 283\\
    \cline{2-5}
    & 30 & 698 & 683 & 669\\
    \cline{2-5}
    & 40 & 1216 & 1197 & 1177\\
    \cline{2-5}
    & 50 & 1956 & 1931 & 1907\\
    \cline{2-5}
    & 60 & 2822 & 2792 & 2763 \\
    \hline
    \multirow{6}{*}{Erdos-Renyi ($p_{edge}$ = 0.6)} & 10 & 50 & 46 & 41\\
    \cline{2-5}
    & 20 & 234 & 225 & 215\\
    \cline{2-5}
    & 30 & 504 & 491 & 475\\
    \cline{2-5}
    & 40 & 960 & 940 & 921\\
    \cline{2-5}
    & 50 & 1504 & 1479 & 1455\\
    \cline{2-5}
    & 60 & 2114 & 2085 & 2055 \\
    \hline
    \multirow{6}{*}{Erdos-Renyi ($p_{edge}$ = 0.4)} & 10 & 36 & 31 & 27\\
    \cline{2-5}
    & 20 & 164 & 154 & 145\\
    \cline{2-5}
    & 30 & 362 & 348 & 333\\
    \cline{2-5}
    & 40 & 586 & 566 & 547\\
    \cline{2-5}
    & 50 & 950 & 925 & 901\\
    \cline{2-5}
    & 60 & 1468 & 1440 & 1409 \\
    \hline
    \end{tabular}
    \label{tab:cx_count}
\end{table*}

In this section we show the effect of our optimization methods on reducing the probability of error and the CNOT count of QAOA for Max-Cut. We first show that our proposed reduction is retained in the post transpilation circuit, which is executed on the quantum hardware. Next, we run our simulation with the noise model for \textit{ibmq\_manhattan} from IBM Quantum; this noise model corresponds to the actual noise in the IBM Quantum Manhattan device which has $65$ qubits and a Quantum Volume of $32$~\cite{ibmquantum}. For our simulation purpose, we have considered Erdos-Renyi graphs, where the probability that an edge exists between two vertices, $p_{edge}$, varies respectively from 0.4 to 1 (complete graph). The choice of Erdos-Renyi graph allows us to study the performance of these proposed methods for various sparsity of graphs.

The circuit that we construct is usually not executed as it is in the IBM Quantum hardware. It undergoes a process called transpilation in which

\begin{enumerate}[(i)]
    \item the gates of the circuit are replaced with one, or a sequence of, basis gates which are actually executed in the quantum hardware. The basis gates of the IBM Quantum devices are \{$CNOT$, $SX$, $X$, $R_z$ and Identity\} \cite{ibmquantum},
    
    \item the circuit is mapped to the underlying connectivity (called the coupling map) of the hardware \cite{bhattacharjee2020survey},
    
    \item the number of gates in the circuit is reduced using logical equivalence \cite{burgholzer2020advanced}.
\end{enumerate}

A natural question, therefore, is whether the reduction in CNOT gates is retained post transpilation. In Table~\ref{tab:cx_count} we show the number of CNOT gates in the post transpilation circuit for the \textit{ibmq\_manhattan} device as the number of vertices is varied from $10 - 60$ for each of the graph family considered. Our results readily show that the proposed optimization in the number of CNOT gates still hold good even in the transpiled circuit. Since the \textit{ibmq\_manhattan} device is a 65-qubits device, we show the results upto 60 qubits, but the results show that the trend will continue for higher qubit devices as well, when they become available.

\begin{figure*}[t]
     \centering
     \begin{subfigure}[b]{0.45\textwidth}
         \centering
         \includegraphics[scale=0.3]{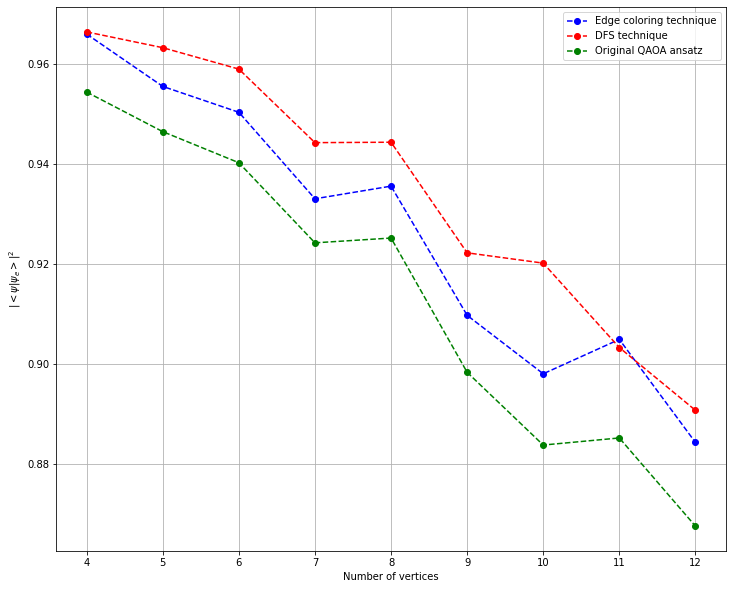}
         \caption{Erdos-Renyi graphs with $p_{edge} = 0.4$}
         \label{col}
     \end{subfigure}
     \hfill
     \begin{subfigure}[b]{0.45\textwidth}
         \centering
         \includegraphics[scale=0.3]{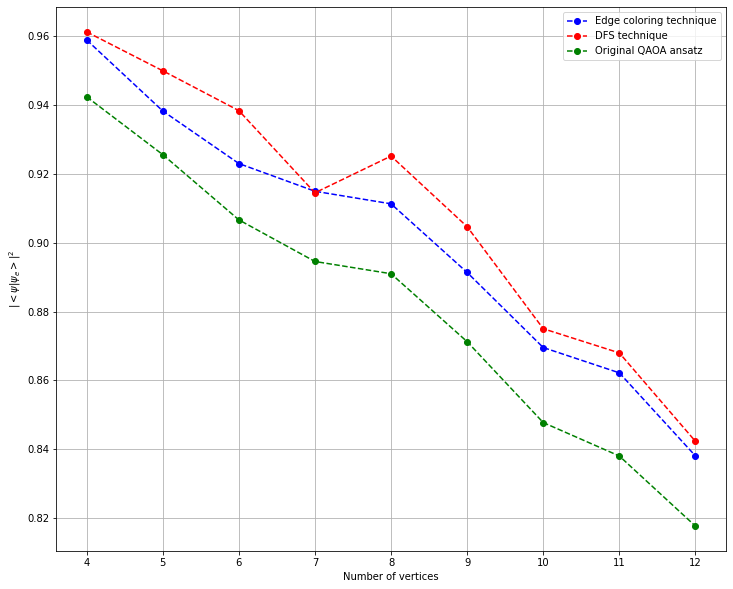}
         \caption{Erdos-Renyi graphs  with $p_{edge} =  0.6$}
         \label{dfs}
     \end{subfigure}
     \newline
     \begin{subfigure}[b]{0.45\textwidth}
         \centering
         \includegraphics[scale=0.3]{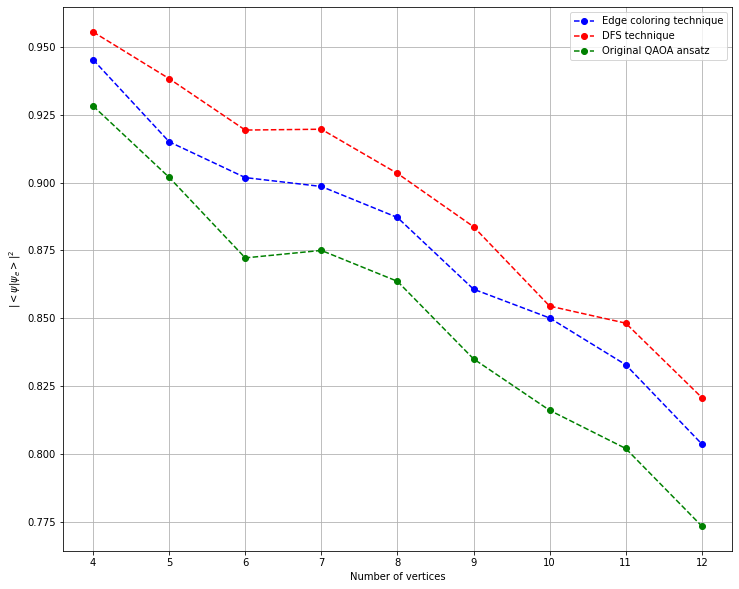}
         \caption{Erdos-Renyi graphs  with $p_{edge} = 0.8$}
         \label{col}
     \end{subfigure}
     \hfill
     \begin{subfigure}[b]{0.45\textwidth}
         \centering
         \includegraphics[scale=0.3]{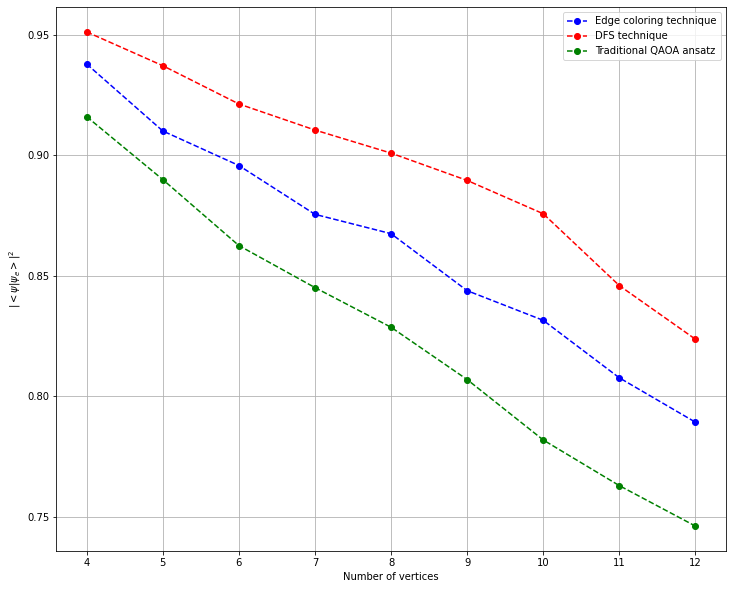}
         \caption{Complete graphs}
         \label{dfs}
     \end{subfigure}
    \caption{$|\braket{\psi|\psi_e}|^2$ for graphs of various sparsity: Erdos Renyi graphs ($p_{edge} = 0.4,~0.6,~ 0.8$) and complete graphs}
    \label{fig:prob_succ}
\end{figure*}

Let $\ket{\psi}$ be the state obtained from the noise-free (ideal) computation of the QAOA circuit, and the state obtained from noisy computation be $\ket{\psi_e}$. The probability of success of the noisy computation, then, is defined as
\begin{equation}
    \label{eq:graph_succ}
    P_{succ} = |\braket{\psi|\psi_e}|^2
\end{equation}
In Fig.~\ref{fig:prob_succ}(a) - (d) we plot $P_{succ}$ of the traditional QAOA ansatz, Edge Coloring based and the DFS based optimization method for Erdos-Renyi graphs, where $p_{edge}$, the probability that an edge exists between two vertices, varies from 0.4 to 1 (complete graph). The choice of Erdos-Renyi graph allows us to study the performance of these proposed methods for various sparsity of graphs. For each case we vary the number of vertices $n$ from 4 to 12. For each value of $n$ and $p_{edge}$, the results are averaged over 20 input graph instances, and each instance is an average of 100 randomly generated noisy circuits by the simulator model for \textit{ibmq\_manhattan} with noise. Our results readily show that the DFS based method outperforms both the Edge Coloring based method and the traditional QAOA in terms of lower error probability.

From Table~\ref{tab:cx_count}, and our simulation results in Fig.~\ref{fig:prob_succ}(a)-(d), we can infer that the DFS based optimization outperforms the Edge Coloring based optimization, which again, outperforms the traditional QAOA in the reduction in CNOT count, and the probability of error in the circuit in (i) the actual transpiled circuit that is executed on the quantum devices, as well as (ii) in realistic noisy scenario of quantum devices.

\section{Conclusion}
\label{sec:con}

In this paper we have proposed two methods to reduce the number of CNOT gates in the traditional QAOA ansatz. The Edge Coloring based method can reduce upto $\lfloor \frac{n}{2} \rfloor$ CNOT gates whereas the DFS based method can reduce $n - 1$ CNOT gates. While the former method provides a depth-optimized circuit, the latter method can increase the depth of the circuit. We analytically derive the constraint for which a particular increase in depth is acceptable given the number of CNOT gates reduced, and show that every graph satisfies this constraint. Therefore, these methods can reduce the number of CNOT gates in the QAOA ansatz for any graph. Finally, we show via simulation, with the \textit{ibmq\_manhattan} noise model, that the DFS based method outperforms the Edge Coloring based method, which in its turn, outperforms the traditional QAOA in terms of lower error probability in the circuit. The transpiler procedure of Qiskit maps a circuit to the underlying hardware connectivity graph, and some gates are reduced in this process. This transpiled circuit is executed on the real hardware. We show, with the \textit{ibmq\_manhattan} coupling map, that the reduction in the number of CNOT gates still holds post transpilation. Therefore, our proposed methods provide a universal way to an improved QAOA ansatz design. On a final note, all the calculations in this paper considers connected graph, but these carry over easily to disconnected graphs as well.

\section*{Acknowledgement}

We acknowledge the use of IBM Quantum services for this work. The views expressed are those of the authors, and do not reflect the official policy or position of IBM or the IBM Quantum team. In this paper we have used the noise model of \textit{ibmq\_manhattan}, which is one of IBM Quantum Hummingbird r2 Processors. 

\section*{Code Availability}

A notebook providing the code to generate the plots of Fig.~\ref{fig:prob_succ}(a)-(d) is available open source at \href{https://github.com/RitajitMajumdar/Optimizing-Ansatz-Design-in-QAOA-for-Max-cut}{https://github.com/RitajitMajumdar/Optimizing-Ansatz-Design-in-QAOA-for-Max-cut}.

\bibliographystyle{unsrt}
\bibliography{main}

\end{document}